\pdfoutput=1
\documentclass[12pt]{article}
\usepackage{custom}
\usepackage{macros}

\title{Lower bounds for set-blocked~clauses proofs}
\author{
  Emre Yolcu\thanks{
    Computer Science Department, Carnegie Mellon University.
    Email: \texttt{eyolcu@cs.cmu.edu}.
  }
}
\date{}

\begin{document}

\maketitle

\begin{abstract}
  We study propositional proof systems with inference rules
  that formalize restricted versions of the ability
  to make assumptions that hold without loss of generality,
  commonly used informally to shorten proofs.
  Each system we study is built on resolution.
  They are called $\BC^-$, $\RAT^-$, $\SBC^-$, and $\GER^-$, denoting respectively
  blocked clauses, resolution asymmetric tautologies,
  set-blocked clauses, and generalized extended resolution%
  ---all ``without new variables.''
  They may be viewed as weak versions of extended resolution~($\ER$)
  since they are defined by first generalizing the extension rule
  and then taking away the ability to introduce new variables.
  Except for~$\SBC^-$, they are known to be
  strictly between resolution and extended resolution.

  Several separations between these systems were proved earlier
  by exploiting the fact that they effectively simulate~$\ER$.
  We answer the questions left open:
  We prove exponential lower bounds for $\SBC^-$~proofs
  of a binary encoding of the pigeonhole principle,
  which separates $\ER$~from~$\SBC^-$.
  Using this new separation, we prove that both $\RAT^-$~and~$\GER^-$
  are exponentially separated from~$\SBC^-$.
  This completes the picture of their relative strengths.
\end{abstract}

\section{Introduction}
\label{sec:introduction}

When writing proofs informally, it is sometimes convenient
to make assumptions that hold ``without loss of generality.''
For instance, if we are proving a statement about real numbers $x$~and~$y$,
we might assume without loss of generality that $x \geq y$
and continue the proof under this additional assumption.
Such an assumption requires justification,
for instance by arguing that the two variables are
interchangeable in the statement being proved.
Assumptions of this kind are not essential to proofs
but they simplify or shorten the presentation.

We study propositional proof systems%
\footnote{Throughout the rest of this paper, by ``proof''
  we mean a proof of unsatisfiability (i.e., a refutation).}
with inference rules that allow making such assumptions.
Extended resolution~\cite{Tse68}
(equivalently, Extended Frege~\cite{CR79})
already simulates this kind of reasoning;
however, it presumably does more,
and its strength is poorly understood.
We thus focus on ``weak'' systems
built on top of resolution~\cite{Bla37,Rob65}
and lacking the ability to introduce any new variables,
while still being able to reason without loss of generality.
Each system relies on a polynomial-time verifiable syntactic condition
to automatically justify the assumption being made,
and the exact form of this condition
determines the strength of the proof system.
The systems are defined by first generalizing the extension rule
and then taking away the ability to introduce new variables,
so, for lack of a better term, we will refer to these systems
collectively as ``weak extended resolution systems'' in this section.
Although we are referring to them as weak,
some variants are surprisingly strong in that they admit
polynomial-size proofs of the pigeonhole principle,
bit pigeonhole principle, parity principle,
clique--coloring principle, and Tseitin tautologies,
as well as being able to undo (with polynomial-size derivations)
the effects of or-ification, xor-ification,
and lifting with indexing gadgets~\cite{BT21}.
In this paper, we study the relative strengths
of several variants of those systems
and answer the questions left open in previous work~\cite[Section~1.4]{YH23}.

\subsection{Motivation}
\label{sec:motivation}

Our interest in weak extended resolution systems
is rooted in two different areas:
proof complexity and satisfiability (SAT) solving.

\subsubsection{Proof complexity}
\label{sec:proof-complexity}

Proof complexity is concerned with the sizes of proofs in propositional proof systems.
The notion of a proof system as accepted in proof complexity is rather general%
---it covers not only the ``textbook'' deductive systems for propositional logic
but also several systems that capture different forms of mathematical reasoning.
For instance, the widely studied proof systems of
cutting planes~\cite{CCT87} and polynomial calculus~\cite{CEI96}
utilize simple forms of geometric and algebraic reasoning, respectively.
Weak extended resolution systems are somewhat similar,
although they do not originate from any specific branch of mathematics.
Instead, they capture the pervasive technique of
reasoning without loss of generality, often used to shorten proofs.
Since proof complexity is concerned with proof size,
the limits of the degree of brevity achievable by this form of reasoning
is a natural question from the perspective of proof complexity.

Moreover, the upper bounds proved by Buss and Thapen~\cite{BT21}
show that many of the usual ``hard'' combinatorial principles
used for proof complexity lower bounds are easy to prove
in certain weak extended resolution systems.
In other words, a modest amount of the ability to reason without loss of generality
lends surprising strength to even a system as weak as resolution,
and the full strength of extended resolution
is not required for the combinatorial principles previously mentioned.
Thus, many of the existing separations between extended resolution
and the commonly studied proof systems can be attributed to the fact
that extended resolution can reason without loss of generality
while the other systems cannot.
Searching for principles that separate extended resolution
from the weak extended resolution systems will help us
better understand other facets of the strength of extended resolution.

\subsubsection{SAT solving}
\label{sec:SAT-solving}

Another motivation for studying the weak extended resolution systems
is their potential usefulness for improvements in SAT solvers,
which are practical implementations of propositional theorem provers
that determine whether a given formula in conjunctive normal form is satisfiable.
When a solver claims unsatisfiability, it is expected to produce a proof
that can be used to verify the claim efficiently.
Modern SAT solvers, which are based on
conflict-driven clause learning (CDCL)~\cite{MS99},
essentially search for resolution proofs.
Consequently, the well-known exponential lower bounds against resolution~%
(e.g.,~\cite{Hak85,Urq87}) imply exponential lower bounds against
the runtimes of CDCL-based solvers.
To overcome the limitations of resolution,
SAT solvers are forced to go beyond CDCL\@.

Many of the current solvers employ ``inprocessing'' techniques~\cite{JHB12},
which support inferences of the kind that we study in this paper.
These techniques are useful in practice;
however, they are implemented as ad hoc additions to CDCL\@.
Weak extended resolution systems hold the potential
for improving SAT solvers in a more principled manner
(e.g., through the development of a solving paradigm
that corresponds to one of those systems
in a manner similar to how CDCL corresponds to resolution).
In comparison, proof systems such as cutting planes, polynomial calculus,
DNF resolution~\cite{Kra01}, or Frege~\cite{CR79}
appear more difficult to take advantage of,
at least in part due to their richer syntax.
When dealing only with clauses, it becomes possible to
achieve highly efficient constraint propagation,
which is an important reason for the speed of CDCL-based solvers.
Extended resolution also works only with clauses;
however, there are currently no widely applicable heuristics
for introducing new variables during proof search.
Weak extended resolution systems are relatively strong
despite using only clauses without new variables.
Thus, those systems are promising for practical proof search algorithms.

Earlier works~\cite{HKSB17,HKB19} showed that solvers based on
certain weak extended resolution systems can automatically discover
small proofs of some formulas that are hard for resolution,
such as the pigeonhole principle and the mutilated chessboard principle.
Still, those solvers fall behind CDCL-based solvers on other classes of formulas.
Moreover, there appears to be a tradeoff when choosing a system for proof search:
stronger systems enable smaller proofs;
however, proof search in such systems is costlier with respect to proof size.
To be able to choose the ideal system for proof search
(e.g., the weakest system that is strong enough for one's purposes),
it is important to understand the relative strengths of the systems in question.
This paper is a step towards that goal.

\subsection{Background}
\label{sec:background}

We briefly review some background,
deferring formal definitions to \cref{sec:preliminaries}.
For comprehensive overviews of related work,
see Buss and Thapen~\cite{BT21} and Yolcu and Heule~\cite{YH23}.

The proof systems we study are based on the notion of ``redundancy.''
A clause is \emph{redundant} with respect to a formula
if it can be added to or removed from the formula
without affecting satisfiability.
Redundancy is a generalization of logical implication:
if $\Gamma \models C$ then the clause $C$
is redundant with respect to the set~$\Gamma$ of clauses;%
\footnote{We use ``set of clauses'' and ``formula'' interchangeably.}
however, the converse is not necessarily true.
When proving unsatisfiability,
deriving redundant clauses corresponds to
making assumptions that hold without loss of generality.%
\footnote{When refuting a formula~$\Gamma$,
  deriving a redundant clause~$C$ may be viewed as stating the following:
  ``If there exists an assignment satisfying~$\Gamma$,
  then there also exists an assignment satisfying both~$\Gamma$~and~$C$,
  so without loss of generality we can assume that $C$ holds.''}
To ensure that proofs can be checked in polynomial time,
we work with restricted versions of redundancy
that rely on syntactic conditions.

Possibly the simplest interesting version
is blockedness~\cite{Kul97,Kul99a}.
We say a clause~$C$ is \emph{blocked}
with respect to a set~$\Gamma$ of clauses
if there exists a literal~$p \in C$
such that all possible resolvents of~$C$ on~$p$
against clauses from~$\Gamma$ are tautological
(i.e., contain a literal and its negation).
Kullmann~\cite{Kul99b} showed that blocked clauses are
special cases of redundant clauses
and thus considered an inference rule that,
given a formula~$\Gamma$, allows us to extend~$\Gamma$
with a clause that is blocked with respect to~$\Gamma$.
This rule, along with resolution, gives
the proof system called \emph{blocked clauses}~($\BC$).
It is apparent from the definition of a blocked clause
that deleting clauses from~$\Gamma$ enlarges
the set of clauses that are blocked with respect to~$\Gamma$.
With this observation, Kullmann defined a strengthening of~$\BC$
called \emph{generalized extended resolution}~($\GER$)
that allows temporary deletion of clauses from~$\Gamma$.
Later works~\cite{JHB12,KSTB18} defined
more general classes of redundant clauses
and proof systems based on them,
called \emph{resolution asymmetric tautologies}~($\RAT$)
and \emph{set-blocked clauses}~($\SBC$).
Both $\RAT$~and~$\SBC$ use relaxed versions of blocked clauses:
$\RAT$ changes the word ``tautological'' in the definition of a blocked clause,
and, in a sense, $\SBC$ considers possible resolvents on more than a single literal.

As defined, $\BC$~simulates extended resolution
since extension clauses can be added in sequence as blocked clauses
if we are allowed to introduce new variables~\cite[see][Section~6]{Kul99b}.
Thus, we disallow new variables to weaken the systems.
A proof of~$\Gamma$ is \emph{without new variables}
if it contains only the variables that already occur in~$\Gamma$.
We denote a proof system variant that disallows new variables
with the superscript~``$-$'' (e.g.,~$\BC^-$ is $\BC$~without new variables).
We are concerned in this paper with the strengths of
the systems $\RAT^-$, $\SBC^-$, and $\GER^-$,
each of which generalizes~$\BC^-$ in different ways.

As a technical side note, the systems we study
are unusual in the following respects:
They are not monotonic, since a clause redundant with respect to~$\Gamma$
is not necessarily redundant with respect to~$\Gamma' \supseteq \Gamma$.
This causes \emph{deletion} (i.e., the ability
to delete clauses in the middle of a proof)
to increase the strength of those systems.
It also requires one to pay attention to
the order of inferences when proving upper bounds.
Additionally, they are not a priori closed under restrictions,
which means that extra care is required when proving lower bounds against them.
More specifically, a proof system~$P$ that simulates~$\BC^-$
is not closed under restrictions unless $P$
also simulates extended resolution~\cite[Theorem~2.4]{BT21}.
It follows from earlier lower bounds~\cite{Kul99b,BT21}
and \cref{sec:lower-bound-for-BPHP} in this paper
that $\BC^-$, $\RAT^-$, $\SBC^-$, and $\GER^-$ are not closed under restrictions.

\subsection{Results}
\label{sec:results}

This work follows up on Yolcu and Heule~\cite{YH23},
which proved separations between the different generalizations of~$\BC^-$
by exploiting the fact that, although the systems cannot introduce new variables,
they nevertheless effectively simulate~\cite{PS10} extended resolution.
Their strategy uses so-called ``guarded extension variables,''
where we consider systems $P$~and~$Q$
that both effectively simulate a strong system~$R$,
and we incorporate extension variables into formulas
in a guarded way that allows~$P$ to simulate an~$R$-proof
while preventing~$Q$ from making any meaningful use
of the extension variables to achieve a speedup.
This allows using, as black-box,
a separation of $R$~from~$Q$ to separate $P$~from~$Q$.
For more details about this strategy,
we refer the reader to Yolcu and Heule~\cite[Section~1.3]{YH23}.

\begin{figure}[!ht]
  \centering
  \begin{tikzpicture}[
    scale=3.15,
    % horizontal center at origin,
    every node/.style={minimum width={width("$\DRAT^{\smash{-}}$")}},
    axis/.style={black!35}
    ]
    \node (BC-) at (0, 0) {$\BC^{\smash{-}}$};
    \node (SBC-) at ($(BC-) + (225:1)$) {$\SBC^{\smash{-}}$};
    \node (SBC-shleft-RAT) at ($(SBC-) + (-0.0225, 0)$) {\phantom{$\SBC^{\smash{-}}$}};
    \node (SBC-shleft-SPR) at ($(SBC-) + (-0.05, 0)$) {\phantom{$\SBC^{\smash{-}}$}};
    \node (RAT-) at (0, 2) {$\RAT^{\smash{-}}$};
    \node (RAT-shleft) at ($(RAT-) + (-0.0225, 0)$) {\phantom{$\RAT^{\smash{-}}$}};
    \node (SPR-) at ($(BC-) + (225:1) + (0, 2)$) {$\SPR^{\smash{-}}$};
    \node (SPR-shleft) at ($(SPR-) + (-0.05, 0)$) {\phantom{$\SPR^{\smash{-}}$}};
    \node (GER-) at (1.1, 0) {$\GER^{\smash{-}}$};
    \node (DBC-) at (2.2, 0) {$\DBC^{\smash{-}}$};
    \node (DSBC-) at ($(BC-) + (225:1) + (2.2, 0)$) {$\DSBC^{\smash{-}}$};
    \node (DRAT-) at (2.2, 2) {$\DRAT^{\smash{-}}$};
    \node (DSPR-) at ($(BC-) + (225:1) + (2.2, 2)$) {$\DSPR^{\smash{-}}$};
    \node[minimum size=3pt, inner sep=0pt] (SBC-GER-) at ($(SBC-)!0.5!(GER-)$) {};
    \node[minimum size=3pt, inner sep=0pt] (SBC-RAT-) at ($(SBC-)!0.5!(RAT-)$) {};
    \node[minimum size=3pt, inner sep=0pt] (SBC-RAT-shleft) at ($(SBC-)!0.5!(RAT-) + (-0.0225, 0)$) {};
    \draw[style=stronger] (DRAT-) -- (RAT-);
    \draw[style=stronger] (SPR-) -- (RAT-);
    \draw[style=equivalent, nontriv] (DBC-) -- (DRAT-);
    \draw[style=stronger, new] (SPR-shleft) -- node[above, sloped, xscale=-1, yscale=-1, yshift=2pt] {\scriptsize\cref{thm:SBC-lower-bound}} (SBC-shleft-SPR);
    \draw[style=incomparable] (RAT-) -- (GER-);
    \draw[style=stronger] (RAT-) -- (BC-);
    \draw[style=stronger] (GER-) -- (BC-);
    \draw[style=stronger] (DBC-) -- (GER-);
    \draw[style=stronger] (SBC-) -- (BC-);
    \draw[style=separated, new] (SBC-GER-) -- node[above, sloped] {\scriptsize\cref{thm:GER-from-SBC}} (SBC-);
    \draw[style=separated, new] (SBC-RAT-shleft) -- node[above, sloped] {\scriptsize\cref{thm:RAT-from-SBC}} (SBC-shleft-RAT);
    \draw[style=separated] (SBC-GER-) -- (GER-);
    \draw[style=separated] (SBC-RAT-shleft) -- (RAT-shleft);
    \draw[style=stronger, overdraw, new] (DSBC-) -- node[above, yshift=2pt] {\scriptsize\cref{thm:SBC-lower-bound}} (SBC-);
    \draw[style=equivalent, overdraw, nontriv] (DBC-) -- (DSBC-);
    \draw[style=simulates, overdraw] (DSPR-) -- (SPR-);
    \draw[style=equivalent, overdraw, nontriv] (DSBC-) -- (DSPR-);
    \draw[style=equivalent, overdraw, nontriv] (DRAT-) -- (DSPR-);
    \draw[style=axis] (DBC-) -- ++(0.5, 0);
    \draw[style=axis] (RAT-) -- ++(0, 0.5);
    \draw[style=axis] (SBC-) -- ++(225:0.5);
  \end{tikzpicture}
  \caption{In the above diagram, the proof systems are placed in
    three-dimensional space with $\BC^-$, the weakest system, at the origin.
    Moving away from the origin along each axis
    corresponds to a particular way of generalizing
    (i.e., strengthening) a proof system.
    The systems prefixed with~``$\mathsf{D}$'' allow
    the arbitrary deletion of a clause as a proof step.
    For systems $P$~and~$Q$, we use
    $P \simulates Q$ to denote that $P$ simulates~$Q$;
    (and $P \ntsimulates Q$ to indicate an ``interesting'' simulation,
    where $P$ is not simply a generalization of~$Q$);
    $P \separated Q$ to denote that $P$ is exponentially separated from~$Q$
    (i.e., there exists an infinite sequence of formulas
    admitting polynomial-size proofs in~$P$
    while requiring exponential-size proofs in~$Q$);
    and $P \stronger Q$ to denote that $P$ both simulates~$Q$
    and is exponentially separated from~$Q$.
    Arrows in \textcolor{customred}{red}
    indicate the relationships that are new in this paper.
    To reduce clutter, some relationships that are implied by transitivity
    are not displayed (e.g., $\DBC^-$ simulates~$\RAT^-$
    and is exponentially separated from it through~$\DRAT^-$).}
  \label{fig:redundancy-landscape}
\end{figure}
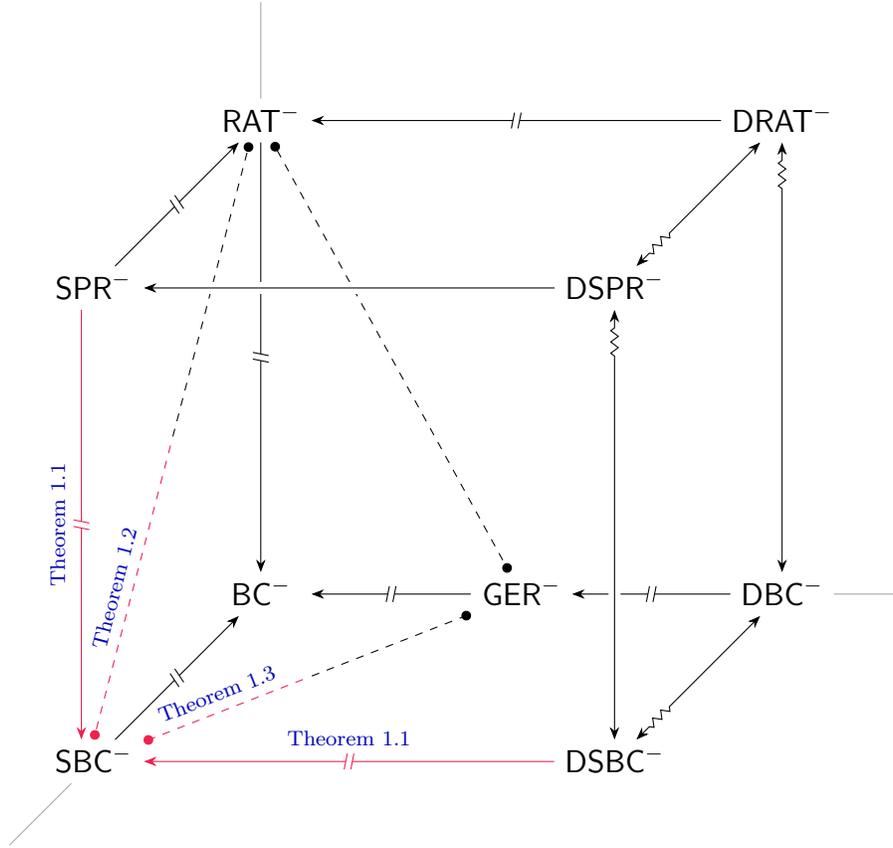

In this paper, we prove the following results, where each formula indexed by~$n$
has $n^{O(1)}$~variables and $n^{O(1)}$~clauses.
\cref{fig:redundancy-landscape} summarizes
the proof complexity landscape around~$\BC^-$ after these results.

We first show exponential lower bounds for $\SBC^-$~proofs
of a binary encoding of the pigeonhole principle
called the ``bit pigeonhole principle,''
defined in \cref{sec:lower-bound-for-BPHP}.
(Note that the usual unary encoding of the pigeonhole principle
admits polynomial-size proofs in~$\SBC^-$~\cite[Lemma~7.1]{YH23}.)

\begin{theorem}\label{thm:SBC-lower-bound}
  The bit pigeonhole principle~$\BPHP_n$
  requires $\SBC^-$~proofs of size~$2^{\Omega(n)}$.
\end{theorem}

We then show, using constructions that incorporate
guarded extension variables into~$\BPHP_n$,
that $\RAT^-$~and~$\GER^-$ are both exponentially separated from~$\SBC^-$.

\begin{theorem}\label{thm:RAT-from-SBC}
  There exists an infinite sequence~$(\Gamma_n)_{n = 1}^\infty$ of unsatisfiable formulas
  such that $\Gamma_n$ admits $\RAT^-$~proofs of size~$n^{O(1)}$
  but requires $\SBC^-$~proofs of size~$2^{\Omega(n)}$.
\end{theorem}

\begin{theorem}\label{thm:GER-from-SBC}
  There exists an infinite sequence~$(\Delta_n)_{n = 1}^\infty$ of unsatisfiable formulas
  such that $\Delta_n$ admits $\GER^-$~proofs of size~$n^{O(1)}$
  but requires $\SBC^-$~proofs of size~$2^{\Omega(n)}$.
\end{theorem}

The above results, along with earlier ones,
completely describe the relative strengths
of the weakest generalizations of~$\BC^-$
along each axis in~\cref{fig:redundancy-landscape}.
For lower bounds, this pushes the frontier
to the system called \emph{set propagation redundancy}~($\SPR^-$)~\cite{HKB20}.
We do not define~$\SPR^-$ formally in this paper,
although it may be thought of intuitively
as combining $\SBC^-$~and~$\RAT^-$.
The upper bounds proved by Buss and Thapen~\cite{BT21}
establish~$\SPR^-$ as an interesting target
for proof complexity lower bounds.
(Note that the binary encoding of the pigeonhole principle
that we use to prove exponential lower bounds for~$\SBC^-$
admits polynomial-size proofs in~$\SPR^-$~\cite[Theorem~4.4]{BT21}.)

\section{Preliminaries}
\label{sec:preliminaries}

We assume that the reader is familiar with propositional logic,
proof complexity, resolution, and extended resolution.
We review some concepts to describe our notation.
For notation we follow Yolcu and Heule~\cite{YH23} from which this section is adapted.

We denote the set of strictly positive integers by~$\NN^+$.
For~$n \in \NN^+$, we let $[n] \coloneqq \{1, \dots, n\}$.
For a sequence~$S = (x_1, \dots, x_n)$,
its \emph{length} is~$n$, which we denote by~$\abs{S}$.

\subsection{Propositional logic}
\label{sec:propositional-logic}

We use $0$~and~$1$ to denote~$\false$ and~$\true$, respectively.
A \emph{literal} is a propositional variable or its negation.
A set of literals is \emph{tautological}
if it contains a pair of complementary literals $x$~and~$\lneg{x}$.
A \emph{clause} is the disjunction of a nontautological set of literals.
We use~$\bot$ to denote the empty clause.
We denote by $\bV$~and~$\bL$ respectively
the sets of all variables and all literals.
A \emph{conjunctive normal form formula} (CNF) is a conjunction of clauses.
Throughout this paper, by ``formula'' we mean a CNF\@.
We identify clauses with sets of literals and formulas with sets of clauses.
In the rest of this section we use $C$,~$D$ to denote clauses
and $\Gamma$,~$\Delta$ to denote formulas.

We say $D$ is a \emph{weakening} of~$C$ if $C \subseteq D$.
We denote by~$\var(\Gamma)$ the set of all the variables occurring in~$\Gamma$.

When we know~$C \cup D$ to be nontautological, we write it as~$C \lor D$.
We write~$C \ldor D$ to indicate a \emph{disjoint disjunction},
where $C$~and~$D$ have no variables in common.
We take the disjunction of a clause and a formula as
\begin{equation*}
  C \lor \Delta \coloneqq \{C \lor D \colon
  D \in \Delta \text{ and } C \cup D \text{ is nontautological}\}.
\end{equation*}

An \emph{assignment}~$\alpha$
is a partial function~$\alpha \colon \bV \rightharpoonup \{0,1\}$,
which also acts on literals by letting $\alpha(\lneg{x}) \coloneqq \lneg{\alpha(x)}$.
We identify~$\alpha$ with the set~$\{p \in \bL \colon \alpha(p) = 1\}$,
consisting of all the literals it satisfies.
For a set~$L$ of literals, we let $\lneg{L} \coloneqq \{\lneg{x} \colon x \in L\}$.
In particular, we use~$\lneg{C}$ to denote the smallest assignment
that falsifies all the literals in~$C$.
We say $\alpha$ \emph{satisfies}~$C$, denoted~$\alpha \models C$,
if there exists some~$p \in C$ such that $\alpha(p) = 1$.
We say $\alpha$ satisfies~$\Gamma$ if for all~$C \in \Gamma$ we have $\alpha \models C$.
For~$C$ that $\alpha$ does not satisfy,
the \emph{restriction} of~$C$ under~$\alpha$ is
$C|_\alpha \coloneqq C \setminus \{p \in C \colon \alpha(p) = 0\}$.
Extending this to formulas, the restriction of~$\Gamma$ under~$\alpha$ is
$\Gamma|_\alpha \coloneqq \{C|_\alpha
\colon C \in \Gamma \text{ and } \alpha \not\models C\}$.

We say $\Gamma$~and~$\Delta$ are \emph{equisatisfiable}, denoted~$\Gamma \eqsat \Delta$,
if they are either both satisfiable or both unsatisfiable.
With respect to~$\Gamma$, a clause~$C$ is \emph{redundant}
if $\Gamma \setminus \{C\} \eqsat \Gamma \eqsat \Gamma \cup \{C\}$.
We sometimes write~$\Gamma \cup \{C\}$ as~$\Gamma \land C$.

\subsection{Proof complexity and resolution}
\label{sec:proof-complexity-and-resolution}

For a proof system~$P$ and a formula~$\Gamma$, we define
\begin{equation*}
  \size_P(\Gamma) \coloneqq \min \{\abs{\Pi}
  \colon \Pi \text{ is a $P$-proof of } \Gamma\}
\end{equation*}
if $\Gamma$ is unsatisfiable and $\size_P(\Gamma) \coloneqq \infty$ otherwise.
A proof system~$P$ \emph{simulates}~$Q$ if every~$Q$-proof
can be converted in polynomial time into a~$P$-proof of the same formula.
Proof systems $P$~and~$Q$ are \emph{equivalent} if they simulate each other.
We say $P$ is \emph{exponentially separated} from~$Q$
if there exists some sequence~$(\Gamma_n)_{n=1}^\infty$ of formulas
such that $\size_P(\Gamma_n) = n^{O(1)}$ while $\size_Q(\Gamma_n) = 2^{\Omega(n)}$.
We call such a sequence of formulas \emph{easy} for~$P$ and \emph{hard} for~$Q$.

Let $C \ldor x$ and $D \ldor \lneg{x}$ be clauses,
where $x$ is a variable, such that the set~$C \cup D$ is nontautological.
We call the clause~$C \lor D$ the \emph{resolvent}
of~$C \lor x$ and~$D \lor \lneg{x}$ on~$x$.
We define a resolution proof in a slightly different form than usual:
as a sequence of formulas instead of a sequence of clauses.%
\footnote{The resulting proof system is equivalent
  to the usual version of resolution.}

\begin{definition}
  A \emph{resolution proof} of a formula~$\Gamma$
  is a sequence~$\Pi = (\Gamma_1, \dots, \Gamma_N$) of formulas
  such that $\Gamma_1 = \Gamma$, $\bot \in \Gamma_N$,
  and for all~$i \in [N - 1]$, we have $\Gamma_{i+1} = \Gamma_i \cup \{C\}$,
  where $C$ is either the resolvent of two clauses in~$\Gamma_i$
  or a weakening of some clause in~$\Gamma_i$.
  The size of~$\Pi$ is~$N$.
\end{definition}

We write~$\Res$ to denote the resolution proof system.
A well known fact is that resolution is closed under restrictions:
if $(\Gamma_1, \Gamma_2, \dots, \Gamma_N)$ is a resolution proof of~$\Gamma$,
then for every assignment~$\alpha$,
the sequence $(\Gamma_1|_\alpha, \Gamma_2|_\alpha, \dots, \Gamma_N|_\alpha)$
contains as a subsequence a resolution proof of~$\Gamma|_\alpha$.
This implies in particular the following.

\begin{lemma}\label{thm:resolution-under-restrictions}
  For every formula~$\Gamma$ and every assignment~$\alpha$,
  $\size_\Res(\Gamma|_\alpha) \leq \size_\Res(\Gamma)$.
\end{lemma}

A \emph{unit propagation proof} is a resolution proof
where each use of the resolution rule
has as at least one of its premises
a clause that consists of a single literal.
Unit propagation is not complete.
With $\Gamma$ a formula and $L = \{p_1, \dots, p_k\}$ a set of literals,
we write~$\Gamma \vdash_1 L$ to denote that
there exists a unit propagation proof
of~$\Gamma \land \lneg{p_1} \land \dots \land \lneg{p_k}$.

We next define \emph{extended resolution}~($\ER$),
which is a strengthening of resolution.

\begin{definition}\label{def:extension-clauses}
  Let $\Gamma$ be a formula and $p$,~$q$ be arbitrary literals.
  Consider a \emph{new} variable~$x$
  (i.e., not occurring in any one of~$\Gamma$,~$p$,~$q$).
  We refer to $\{\lneg{x} \lor p,\ \lneg{x} \lor q,\ x \lor \lneg{p} \lor \lneg{q}\}$
  as a set of \emph{extension clauses} for~$\Gamma$.
  In this context, we call~$x$ the \emph{extension variable}.
\end{definition}

\begin{definition}
  A formula~$\Lambda$ is an \emph{extension} for a formula~$\Gamma$
  if there exists a sequence~$(\lambda_1, \dots, \lambda_t)$
  such that $\Lambda = \bigcup_{i = 1}^t \lambda_i$,
  and for all~$i \in [t]$, we have that $\lambda_i$ is a set of
  extension clauses for~$\Gamma \cup \bigcup_{j = 1}^{i - 1} \lambda_j$.
\end{definition}

\begin{definition}
  An \emph{extended resolution proof} of a formula~$\Gamma$
  is a pair~$(\Lambda, \Pi)$,
  where $\Lambda$ is an extension for~$\Gamma$
  and $\Pi$ is a resolution proof of~$\Gamma \cup \Lambda$.
  The size of~$(\Lambda, \Pi)$ is~$\abs{\Lambda} + \abs{\Pi}$.
\end{definition}

\subsection{Redundancy criteria}
\label{sec:redundancy-criteria}

We recall the syntactic redundancy criteria that lead to the inference rules we study.
The definitions are taken from Yolcu and Heule~\cite[Section~3]{YH23},
which in turn adapted them from previous works~\cite{Kul99b,JHB12,KSTB18,HKB20,BT21}.

\begin{definition}\label{def:blocked-clause}
  A clause~$C = p \ldor C'$ is a \emph{blocked clause} (BC)
  for the literal~$p$ with respect to a formula~$\Gamma$
  if, for every clause~$D$ of the form~$\lneg{p} \ldor D'$ in~$\Gamma$,
  the set~$C' \cup D'$ is tautological.
\end{definition}

\begin{definition}\label{def:resolution-asymmetric-tautology}
  A clause~$C = p \ldor C'$ is a \emph{resolution asymmetric tautology} (RAT)
  for the literal~$p$ with respect to a formula~$\Gamma$
  if, for every clause~$D$ of the form~$\lneg{p} \ldor D'$ in~$\Gamma$,
  we have $\Gamma \vdash_1 C' \cup D'$.
\end{definition}

\begin{definition}\label{def:set-blocked-clause}
  A clause~$C$ is a \emph{set-blocked clause} (SBC)
  for a nonempty~$L \subseteq C$ with respect to a formula~$\Gamma$
  if, for every clause~$D \in \Gamma$
  with~$D \cap \lneg{L} \neq \varnothing$ and~$D \cap L = \varnothing$,
  the set~$(C \setminus L) \cup (D \setminus \lneg{L})$ is tautological.
\end{definition}

We say $C$ is a BC with respect to~$\Gamma$
if there exists a literal~$p \in C$
for which $C$ is a BC with respect to~$\Gamma$,
and similarly for RAT and SBC\@.
It was shown in previous works~\cite{Kul99b,JHB12,KSTB18}
that BCs, RATs, and SBCs are redundant,
which makes it possible to use them to define proof systems.

\begin{definition}
  A \emph{blocked clauses proof} of a formula~$\Gamma$
  is a sequence~$\Pi = (\Gamma_1, \dots, \Gamma_N$) of formulas
  such that $\Gamma_1 = \Gamma$, $\bot \in \Gamma_N$,
  and for all~$i \in [N - 1]$, we have $\Gamma_{i+1} = \Gamma_i \cup \{C\}$,
  where $C$ is either the resolvent of two clauses in~$\Gamma_i$,
  a weakening of some clause in~$\Gamma_i$,
  or a blocked clause with respect to~$\Gamma_i$.
  The size of~$\Pi$ is~$N$.
\end{definition}

We write~$\BC$ to denote the blocked clauses proof system.
Replacing ``blocked clause'' by ``resolution asymmetric tautology''
in the above definition gives the $\RAT$~proof system.
Replacing it by ``set-blocked clause'' gives the $\SBC$~proof system.%
\footnote{For an $\SBC$~proof to be polynomial-time verifiable,
  every step in the proof that adds a clause~$C$ as set-blocked
  is expected to indicate the subset~$L \subseteq C$ for which $C$ is set-blocked.
  With that said, we leave this requirement out of our definitions to reduce clutter.}
$\RAT$~and~$\SBC$ are two generalizations of~$\BC$, and we now define another,
called \emph{generalized extended resolution}~($\GER$),
which reduces the dependence of the validity of~$\BC$ inferences
on the order of clause additions~\cite[see][Section~1.3]{Kul99b}.
We need to introduce the concept of a blocked extension%
\footnote{Instead of the original definition
  of a blocked extension~\cite[Definition~6.3]{Kul99b},
  we use a convenient characterization~\cite[Lemma~6.5]{Kul99b},
  which is simpler to state, as the definition.}
before we can proceed with the definition of~$\GER$.

\begin{definition}\label{def:blocked-extension}
  A formula~$\Lambda$ is a \emph{blocked extension} for a formula~$\Gamma$
  if there exists a subset~$\Gamma'$ of~$\Gamma$
  and an ordering~$(C_1, \dots, C_r)$
  of all the clauses in~$\Lambda \cup (\Gamma \setminus \Gamma')$
  such that for all~$i \in [r]$ the clause~$C_i$ is blocked
  with respect to~$\Gamma' \cup \bigcup_{j = 1}^{i - 1} \{C_j\}$.
\end{definition}

\begin{definition}
  A \emph{generalized extended resolution proof} of a formula~$\Gamma$
  is a pair~$(\Lambda, \Pi)$,
  where $\Lambda$ is a blocked extension for~$\Gamma$
  and $\Pi$ is a resolution proof of~$\Gamma \cup \Lambda$.
  The size of~$(\Lambda, \Pi)$ is~$\abs{\Lambda} + \abs{\Pi}$.
\end{definition}

Note that the definitions in this section do not prohibit
BCs, RATs, or SBCs with respect to~$\Gamma$
from containing variables not occurring in~$\Gamma$.
We study the variants of $\BC$, $\RAT$, $\SBC$, and $\GER$
that disallow the use of new variables.
A proof of~$\Gamma$ is \emph{without new variables}
if all the variables occurring in the proof are in~$\var(\Gamma)$.
In the case of~$\GER$, this constraint applies to
both the blocked extension and the resolution part:
a proof~$(\Lambda, \Pi)$ of~$\Gamma$ is without new variables
if all the variables occurring in~$\Lambda$~or~$\Pi$ are in~$\var(\Gamma)$.
We use $\BC^-$, $\RAT^-$, $\SBC^-$, and $\GER^-$
to denote the variants without new variables.

\section{Lower bound for the bit pigeonhole principle}
\label{sec:lower-bound-for-BPHP}

Let $n = 2^k$, with~$k \in \NN^+$.
For a propositional variable~$v$,
let us write~$v \neq 0$ and~$v \neq 1$
to denote the literals $v$~and~$\lneg{v}$, respectively.
The \emph{bit pigeonhole principle} is the contradiction stating that
each pigeon in~$[n + 1]$ can be assigned a distinct binary string from~$\{0,1\}^k$,
where we identify strings with holes.
For each pigeon~$x \in [n + 1]$, the variables~$p^x_1, \dots, p^x_k$
represent the bits of the string assigned to~$x$.
More formally, we write the bit pigeonhole principle as
\begin{equation*}
  \BPHP_n \coloneqq
  \bigcup_{\substack{x, y \in [n + 1],\; x \neq y \\ (h_1, \dots, h_k) \in \{0, 1\}^k}}
  \mleft\{\bigvee_{\ell = 1}^k p^x_\ell \neq h_\ell
  \lor \bigvee_{\ell = 1}^k p^y_\ell \neq h_\ell \mright\},
\end{equation*}
which asserts that for all~$x, y \in [n + 1]$ such that $x \neq y$,
the binary strings~$p^x_1 \dots p^x_k$ and~$p^y_1 \dots p^y_k$ are different.%
\footnote{In our asymptotic results that use the bit pigeonhole principle,
  it is tacitly understood that $\BPHP_n$ could be defined for every integer~$n \geq 2$
  (as opposed to only powers of two) by letting $\BPHP_n$ be identical to~$\BPHP_m$,
  where $m$ is the largest power of two not exceeding~$n$.}
We denote by~$P_x$ the set~$\{p^x_1, \dots, p^x_k,
\lneg{p^x_1}, \dots, \lneg{p^x_k}\}$ of all the literals concerning pigeon~$x$.

For a set~$L$ of literals,
its \emph{pigeon-width} is the number of distinct pigeons it mentions,
where a pigeon~$x \in [n + 1]$ is \emph{mentioned} if there exists some~$\ell \in [k]$
such that some literal of the variable~$p^x_\ell$ is in~$L$.
We write~$L(x)$ to denote the set~$L \cap P_x$.
In other words, $L(x)$ is the largest subset of~$L$ that mentions only the pigeon~$x$.

Before proceeding with the $\SBC^-$~lower bound for the bit pigeonhole principle,
we observe the result below, which will also be useful later.
It is deduced in a straightforward way
from the definition of a set-blocked clause,
using the following fact:
if a clause~$C$ is an SBC with respect to a formula~$\Gamma$,
then $C$ is an SBC with respect to every subset of~$\Gamma$.
(Note that a similar result does not necessarily hold for~$\RAT^-$.)

\begin{lemma}\label{thm:WLOG-SBC-preprocesses}
  Without loss of generality,
  all of the set-blocked clause additions in an $\SBC^-$~proof
  are performed before any resolution or weakening steps.
\end{lemma}

\cref{thm:WLOG-SBC-preprocesses} allows us to reduce
$\SBC^-$~lower bounds for a formula~$\Gamma$
to resolution lower bounds for another formula~$\Gamma \cup \Sigma$,
where $\Sigma$ is a set of clauses derivable from~$\Gamma$
by a sequence of set-blocked clause additions without new variables.

A common strategy for proving resolution lower bounds is to
first show that every proof contains some ``complex'' clause
(in our case, a clause of large pigeon-width),
and then argue that the existence of a small proof
implies the existence of another proof where no clause is complex.
The second step typically involves restricting
the clauses of the proof under a suitable assignment.
We work with assignments that correspond to
partial matchings of pigeons to holes,
as in the case of the $\RAT^-$~lower bound
by Buss and Thapen~\cite[Section~5]{BT21}.

We say an assignment~$\rho$ that sets some variables of~$\BPHP_n$
is a \emph{partial matching} if $\rho$~sets
all of the bits for the pigeons it mentions
in such a way that no two pigeons are in the same hole,
thus representing a matching of pigeons to holes.
To prove $\SBC^-$~lower bounds for~$\BPHP_n$,
we will need the following pigeon-width lower bound
for resolution proofs of restrictions of~$\BPHP_n$ under partial matchings,
which is established by a straightforward Adversary strategy
in the Prover--Adversary game~\cite{Pud00}.
We define the pigeon-width of a proof
as the maximum pigeon-width of any clause in the proof.

\begin{lemma}[{\cite[Lemma~5.2]{BT21}}]\label{thm:resolution-width-lower-bound-for-BPHP}
  Let $\rho$ be a partial matching of $m$~pigeons to holes.
  Then every resolution proof of~$(\BPHP_n)|_\rho$
  has pigeon-width at least~$n - m$.
\end{lemma}

We will additionally need a pigeon-width lower bound for
set-blocked clauses (without new variables)
with respect to~$\BPHP_n$, which follows from a simple inspection.

\begin{lemma}\label{thm:SBC-width-lower-bound-for-BPHP}
  Every set-blocked clause with respect to~$\BPHP_n$
  that is without new variables has pigeon-width~$n + 1$.
\end{lemma}
\begin{proof}
  Let $C = L \ldor C'$ be a set-blocked clause for~$L$
  with respect to~$\BPHP_n$ that is without new variables.
  Let $x$ be a pigeon mentioned in~$L$.
  Such a pigeon exists since $L$ is nonempty.
  Let $y$ be a pigeon different from~$x$.
  We claim that $C$ mentions~$y$.

  Let $D \in \BPHP_n$ be a clause
  that contains~$\lneg{L(x)} \cup C'(x)$ and mentions~$y$.
  Such a clause exists since $\lneg{L(x)} \cup C'(x)$ is
  simply a nontautological subset of~$P_x$
  and, by the definition of~$\BPHP_n$,
  each such subset is contained in
  some clause in~$\BPHP_n$ that mentions~$y$.
  Note that every clause in~$\BPHP_n$ mentions exactly two pigeons;
  in particular, the clause~$D$ mentions only the pigeons $x$~and~$y$.

  Since $D$ is a clause, it is nontautological.
  As a consequence, $D \cap L(x)$ is empty
  and $C'(x) \cup (D \setminus \lneg{L})$ is nontautological.
  Then, since $C$ is set-blocked for~$L$ with respect to~$\BPHP_n$,
  either $D \cap (L \setminus L(x))$ is nonempty
  or $(C' \setminus C'(x)) \cup (D \setminus \lneg{L})$ is tautological.
  Now, neither of $L \setminus L(x)$~and~$C' \setminus C'(x)$ mentions~$x$.
  Since $D$ mentions only the pigeons $x$~and~$y$,
  the pigeon~$y$ must be mentioned by~$L$ in the former case and~$C'$ in the latter.
  Either way, $C$ mentions~$y$.
\end{proof}

\begin{theorem}\label{thm:SBC-size-lower-bound-for-BPHP}
  The formula~$\BPHP_n$ requires exponential-size proofs in~$\SBC^-$.
\end{theorem}
\begin{proof}
  Let $\Pi$ be an $\SBC^-$~proof of~$\BPHP_n$ of size~$N$.
  By \cref{thm:WLOG-SBC-preprocesses}, we may view~$\Pi$
  as a resolution proof of the formula~$\BPHP_n \cup \Sigma$,
  where $\Sigma$ is a set of clauses derivable from~$\BPHP_n$
  by a sequence of set-blocked clause additions without new variables.
  We will show by the probabilistic method that if $N < 2^{n / 64}$,
  then there exists a partial matching~$\rho$ of $n / 2$~pigeons to holes
  such that $(\BPHP_n)|_\rho$ has a resolution proof
  of pigeon-width strictly less than~$n / 2$.

  Let $R$ be a random partial matching constructed by
  choosing a random pigeon and assigning it to a random available hole
  until $n / 2$~pigeons are matched to holes.
  We denote by $r_i$ the random assignment
  performed at the $i$th~step of this process:
  if pigeon~$x$ was assigned to hole~$(h_1, \dots, h_k)$,
  then $r_i(p^x_\ell) = h_\ell$ for all~$\ell \in [k]$.

  We say a clause is \emph{wide} if it has pigeon-width at least~$n / 2$.
  Let $C$ be a wide clause.
  Let $x$ denote the $i$th pigeon chosen when constructing~$R$, with $i \leq n / 4$.
  The probability that $x$ is mentioned by~$C$ is at least
  \begin{equation*}
    \frac{n/2 - (n/4 - 1)}{n + 1} \geq 1/4.
  \end{equation*}
  Suppose that $x$ is mentioned by~$C$ through a literal~$p$.
  When $x$ is about to be assigned to a hole, there are
  at least $n / 2 - (n / 4 - 1) \geq n / 4$ available ones
  that would result in~$r_i$ satisfying~$p$.
  Therefore, the conditional probability that $r_i$ satisfies~$C$
  given that the assignments~$r_1, \dots, r_{i - 1}$ do not satisfy~$C$
  is at least~$1 / 16$.
  As a result,
  \begin{equation*}
    \Pr[R \not\models C] < (1 - 1 / 16)^{n / 4} \leq 2^{-n / 64}.
  \end{equation*}

  Suppose that $N < 2^{n / 64}$.
  Let $\Delta$ be the set of all the wide clauses appearing in~$\Pi$.
  By the union bound, $\Pr[R \not\models \Delta] < 1$.
  Thus, there exists a partial matching~$\rho$
  of $n / 2$~pigeons to holes such that $\rho \models \Delta$.
  Also, observe that $\rho \models \Sigma$
  because we have $\Sigma \subseteq \Delta$
  by \cref{thm:SBC-width-lower-bound-for-BPHP}.

  Since resolution is closed under restrictions,
  when we restrict the proof~$\Pi$ under~$\rho$ we obtain
  a resolution proof of~$(\BPHP_n \cup \Sigma)|_\rho = (\BPHP_n)|_\rho$
  without any wide clauses, which contradicts
  \cref{thm:resolution-width-lower-bound-for-BPHP}.
\end{proof}

\section{Separations using guarded extension variables}
\label{sec:separations-using-guarded-extension-variables}

From this point on, given a formula~$\Gamma$, we use $(\Lambda, \Pi)$
to denote the minimum-size $\ER$~proof of~$\Gamma$,%
\footnote{We refer to \emph{the} minimum-size proof
  with the assumption of having fixed some way
  of choosing a proof among those with minimum size.}
where $\Lambda$ is the union of a sequence of
$t(\Gamma) \coloneqq \abs{\Lambda}/3$~sets of extension clauses
such that the $i$th~set~$\lambda_i$ is of the form
$\{\lneg{x_i} \lor p_i,\ \lneg{x_i} \lor q_i,\ x_i \lor \lneg{p_i} \lor \lneg{q_i}\}$.
We thus reserve~$\left\{x_1, \dots, x_{t(\Gamma)}\right\}$
as the set of extension variables used in~$\Lambda$.
We assume without loss of generality
that the variables of $p_i$~and~$q_i$ are
in~$\var(\Gamma) \cup \{x_1, \dots, x_{i-1}\}$ for all~$i \in [t(\Gamma)]$.

\subsection{Separation of~\psnnv{RAT} from~\psnnv{SBC}}
\label{sec:RAT-from-SBC}

Let $\Gamma$ be a formula and $(\Lambda, \Pi)$ be
the minimum-size $\ER$~proof of~$\Gamma$ as described above.
Consider the transformation
\begin{equation}\label{eq:unsatisfiable-projection-transformation}
  \cG(\Gamma) \coloneqq \Gamma
  \cup \bigcup_{i = 1}^{t(\Gamma)} \bigl[
  (x_i \lor \Gamma)
  \cup (\lneg{x_i} \lor \Gamma)
  \bigr],
\end{equation}
where $x_1, \dots, x_{t(\Gamma)}$ are the extension variables used in~$\Lambda$.

It becomes possible to prove~$\cG(\Gamma)$ in~$\RAT^-$
by simulating the $\ER$~proof of~$\Gamma$
using the extension variables present in the formula,
resulting in the following.

\begin{lemma}[{\cite[Lemma~5.1]{YH23}}]\label{thm:projection-easy-for-RAT}
  For every formula~$\Gamma$,
  $\size_{\RAT^-}(\cG(\Gamma)) \leq \size_{\ER}(\Gamma)$.
\end{lemma}

For~$\SBC^-$, the formula~$\cG(\Gamma)$ is at least as hard as~$\Gamma$.
We need the following definition before we prove this fact.

\begin{definition}\label{def:projection}
  The \emph{projection} of a formula~$\Gamma$ onto a literal~$p$ is the formula
  \begin{equation*}
    \proj_p(\Gamma) \coloneqq \{C \setminus \{p\} \colon
    C \in \Gamma \text{ and } p \in C\}.
  \end{equation*}
\end{definition}

Our main tool in proving $\SBC^-$~lower bounds against
the constructions that incorporate guarded extension variables
is a version of the following characterization of blocked clauses,
which was already observed by Kullmann~\cite[see][Section~4]{Kul99b}.

\begin{lemma}[{\cite[Lemma~3.15]{YH23}}]\label{thm:BC-characterization-projection}
  A clause~$C = p \ldor C'$ is a BC for~$p$ with respect to a formula~$\Gamma$
  if and only if the assignment~$\lneg{C'}$ satisfies~$\proj_{\lneg{p}}(\Gamma)$.
\end{lemma}

\cref{thm:BC-characterization-projection} suffices for proving
$\BC^-$~lower bounds against~$\cG(\Gamma)$~\cite[see][Lemma~5.2]{YH23};
however, for $\SBC^-$~lower bounds we need the following result.

\begin{lemma}\label{thm:SBC-projection-onto-literal}
  If a clause~$C = L \ldor C'$ is an SBC
  for~$L$ with respect to a formula~$\Gamma$,
  then for every~$p \in L$, the assignment~$L \cup \lneg{C'}$
  satisfies~$\proj_{\lneg{p}}(\Gamma)$.
\end{lemma}
\begin{proof}
  Suppose that $C = L \ldor C'$ is an SBC for~$L$ with respect to~$\Gamma$,
  let $p \in L$, and let $D' \in \proj_{\lneg{p}}(\Gamma)$.
  Then the clause~$D = \lneg{p} \ldor D'$ is in~$\Gamma$.
  Note that $D \cap \lneg{L}$ is nonempty.
  Then, since $C$ is an SBC for~$L$, either $D \cap L$ is nonempty
  or $C' \cup (D \setminus \lneg{L})$ is tautological.

  \begin{description}
  \item[Case 1:] $D \cap L \neq \varnothing$.
    Since $\lneg{p} \notin L$, the set~$D' \cap L$ also is nonempty;
    therefore $L \models D'$.
  \item[Case 2:] \textit{$C' \cup (D \setminus \lneg{L})$ is tautological.}
    Since neither of $C'$~and~$D \setminus \lneg{L}$ is tautological,
    their union is tautological if and only if
    $\lneg{C'} \cap (D \setminus \lneg{L})$ is nonempty.
    This implies in particular that $\lneg{C'} \cap D'$ is nonempty;
    therefore $\lneg{C'} \models D'$. \qedhere
  \end{description}
\end{proof}

\cref{thm:SBC-projection-onto-literal} implies that
if the projection of a formula onto a literal~$p$ is unsatisfiable,
then, with respect to the formula, no clause is set-blocked
for any set that contains~$\lneg{p}$.

The intuition behind the construction of~$\cG(\Gamma)$ is as follows.
We incorporate the extension variables
into the formula while having~$\Gamma$ be
the projection for each added literal.
Thus, if $\Gamma$ is unsatisfiable,
we render the extension variables useless
in set-blocked clause additions with respect to~$\cG(\Gamma)$
while still allowing~$\RAT^-$ to take advantage of them.
In particular, it becomes unnecessary for a set-blocked clause~$C$
with respect to~$\cG(\Gamma)$ to include
any of the extension variables present in the formula.
This is because every such clause~$C$ has some subset~$C'$
without any of the extension variables
that is still set-blocked with respect to~$\cG(\Gamma)$.
Moreover, since $\Gamma \subseteq \cG(\Gamma)$, the clause~$C'$
is set-blocked also with respect to~$\Gamma$.
The alternative way to use the extension variables in~$\cG(\Gamma)$
is to derive~$x_i$ from~$x_i \lor \Gamma$,
but this involves proving~$\Gamma$.
When $\Gamma$ is hard for~$\SBC^-$, we leave no way for~$\SBC^-$
to make any meaningful use of the extension variables to achieve a speedup.
In the end, an $\SBC^-$~proof of~$\cG(\Gamma)$ might as well
ignore the extension variables present in the formula,
falling back to an $\SBC^-$~proof of~$\Gamma$.

\begin{lemma}\label{thm:projection-hard-for-SBC}
  For every formula~$\Gamma$,
  $\size_{\SBC^-}(\cG(\Gamma)) \geq \size_{\SBC^-}(\Gamma)$.
\end{lemma}
\begin{proof}
  When $\Gamma$ is satisfiable, the inequality holds trivially,
  so suppose that $\Gamma$ is unsatisfiable.

  Suppose that $\cG(\Gamma)$ has an $\SBC^-$~proof of size~$N$.
  By \cref{thm:WLOG-SBC-preprocesses}, we may view such a proof
  as a resolution proof of the formula~$\cG(\Gamma) \cup \Sigma$,
  where $\Sigma$ is a set of clauses derivable from~$\cG(\Gamma)$
  by a sequence of set-blocked clause additions without new variables.
  Let $X = \left\{x_1, \dots, x_{t(\Gamma)}\right\}$
  denote the set of extension variables incorporated into~$\cG(\Gamma)$,
  and consider an assignment~$\alpha$ defined as
  \begin{equation*}
    \alpha(v) =
    \begin{cases}
      1 & \text{if } v \in X \\
      \text{undefined} & \text{otherwise}.
    \end{cases}
  \end{equation*}
  By \cref{thm:resolution-under-restrictions}, there exists a resolution proof
  of the formula~$(\cG(\Gamma) \cup \Sigma)|_\alpha = \Gamma \cup \Sigma|_\alpha$
  of size at most~$N - \abs{\Sigma}$.
  We claim that the clauses in~$\Sigma|_\alpha$ can be derived
  in sequence from~$\Gamma$ by set-blocked clause additions,
  which implies that there exists
  an $\SBC^-$~proof of~$\Gamma$ of size at most~$N$.

  Let $S = (C_1, \dots, C_r)$ be the ordering in which
  the clauses of~$\Sigma$ are derived from~$\cG(\Gamma)$.
  We will show that if we restrict each clause in~$S$ under~$\alpha$
  and remove the satisfied clauses, then the remaining sequence of clauses
  can be derived from~$\Gamma$ in the same order by set-blocked clause additions.
  More specifically, the goal is to prove that for all~$i \in [r]$
  such that $\alpha$ does not satisfy~$C_i$,
  the clause~$C_i|_\alpha$ is set-blocked
  with respect to~$\Gamma \cup \Phi_{i - 1}|_\alpha$, where
  \begin{equation*}
    \Phi_{i - 1} \coloneqq \bigcup_{j \in [i - 1]} \{C_j\}.
  \end{equation*}

  Let $i \in [r]$, and consider the clause~$C_i$,
  which we write as~$C$ from this point on.
  Suppose that $\alpha$ does not satisfy~$C$,
  so the variables from~$X$ can occur only negatively in~$C$.
  Let $L \subseteq C$ be a subset for which $C$ is set-blocked
  with respect to~$\cG(\Gamma) \cup \Phi_{i - 1}$.
  We will prove that $C|_\alpha$ is set-blocked for~$L|_\alpha$
  with respect to both $\Gamma$~and~$\Phi_{i - 1}|_\alpha$,
  which, by the definition of a set-blocked clause,
  implies that $C|_\alpha$ is set-blocked
  with respect to~$\Gamma \cup \Phi_{i - 1}|_\alpha$.

  Before proceeding, observe that $L$ cannot contain any variables from~$X$:
  If some~$\lneg{x_i}$ is in~$L$, then the assignment~$L \cup \lneg{(C \setminus L)}$
  satisfies~$\proj_{x_i}(\cG(\Gamma)) = \Gamma$
  by \cref{thm:SBC-projection-onto-literal}.
  Since $\Gamma$ is unsatisfiable, no such assignment exists.
  Therefore, $L$ cannot contain~$\lneg{x_i}$,
  which implies that $L|_\alpha = L$.

  \begin{description}
  \item\hspace{-\labelsep}\ignorespaces
    \textit{$C|_\alpha$ is set-blocked
      for~$L$ with respect to~$\Gamma$}:
    Since $\Gamma \subseteq \cG(\Gamma)$, the clause~$C$
    is set-blocked for~$L$ in particular with respect to~$\Gamma$.
    Noting that the variables from~$X$ do not occur in~$\Gamma$,
    we conclude that $C|_\alpha$ also is set-blocked
    for~$L$ with respect to~$\Gamma$.
  \item\hspace{-\labelsep}\ignorespaces
    \textit{$C|_\alpha$ is set-blocked
      for~$L$ with respect to~$\Phi_{i - 1}|_\alpha$}:
    Consider an arbitrary~$D' \in \Phi_{i - 1}|_\alpha$,
    which is the restriction under~$\alpha$
    of some clause~$D \in \Phi_{i - 1}$ that $\alpha$ does not satisfy.
    Suppose $D' \cap \lneg{L} \neq \varnothing$ and $D' \cap L = \varnothing$.
    We need to show that $(C|_\alpha \setminus L)
    \cup (D' \setminus \lneg{L})$ is tautological.

    Since $D' \subseteq D$, we immediately have $D \cap \lneg{L} \neq \varnothing$.
    Now, recall that the variables from~$X$ do not occur in~$L$,
    and observe that $D'$ is simply $D$~with the variables from~$X$ removed.
    We thus have $D \cap L = \varnothing$.
    Then, because $C$ is set-blocked for~$L$ with respect to~$\Phi_{i - 1}$,
    the set~$E = (C \setminus L) \cup (D \setminus \lneg{L})$ must be tautological.
    A variable that occurs both positively and negatively in~$E$
    cannot be from~$X$, since in that case $\alpha$ would satisfy $C$~or~$D$.
    Therefore, the set~$(C|_\alpha \setminus L)
    \cup (D' \setminus \lneg{L})$ also is tautological. \qedhere
  \end{description}
\end{proof}

Invoking \crefnosort{thm:projection-easy-for-RAT,thm:projection-hard-for-SBC}
with $\Gamma$~as the bit pigeonhole principle gives us the separation.

\begin{theorem}\label{thm:separation-RAT-from-SBC}
  The formula~$\cG(\BPHP_n)$ admits polynomial-size proofs in~$\RAT^-$
  but requires exponential-size proofs in~$\SBC^-$.
\end{theorem}
\begin{proof}
  Buss and Thapen~\cite[Theorem~4.4]{BT21}
  gave polynomial-size proofs of~$\BPHP_n$
  in~$\SPR^-$, which $\ER$ simulates.\footnote{
    It is also possible to deduce the existence of polynomial-size $\ER$~proofs
    of~$\BPHP_n$ from the fact that the pigeonhole principle~($\PHP_n$)
    is easy for~$\ER$~\cite{Coo76}, combined with the observation
    that $\PHP_n$ can be derived from~$\BPHP_n$ in polynomial size in~$\ER$.}
  By \cref{thm:projection-easy-for-RAT}, we have
  \begin{equation*}
    \size_{\RAT^-}(\cG(\BPHP_n)) = n^{O(1)}.
  \end{equation*}
  \crefnosort{thm:SBC-size-lower-bound-for-BPHP,thm:projection-hard-for-SBC} give
  \begin{equation*}
    \size_{\SBC^-}(\cG(\BPHP_n)) = 2^{\Omega(n)}.
  \end{equation*}
  Thus, the bit pigeonhole principle with $\cG$~applied to it
  exponentially separates~$\RAT^-$ from~$\SBC^-$.
\end{proof}

\subsection{Separation of~\psnnv{GER} from~\psnnv{SBC}}
\label{sec:GER-from-SBC}

We proceed in a similar way to the previous section.
Let $\Gamma$ be a formula and $(\Lambda, \Pi)$ be
the minimum-size $\ER$~proof of~$\Gamma$.
Let $m \in \NN^+$, and let
\begin{equation*}
  \left\{y_1, \dots, y_m, z_1, \dots, z_m\right\}
  \subseteq \bV \setminus \var(\Gamma \cup \Lambda)
\end{equation*}
be a set of $2m$~distinct variables.
Consider
\begin{align}\label{eq:blocked-pair-with-unsatisfiable-projection-transformation}
  \begin{split}
    V_m(\Gamma)
    &\coloneqq \bigcup_{i = 1}^{t(\Gamma)} \bigcup_{j = 1}^m
      \{x_i \lor y_j \lor \lneg{z_j},\ \lneg{x_i} \lor y_j \lor \lneg{z_j}\}, \\
    W_m(\Gamma)
    &\coloneqq \bigcup_{j = 1}^m \bigl[
      \{\lneg{y_j} \lor z_j\}
      \cup (y_j \lor \Gamma)
      \cup (\lneg{z_j} \lor \Gamma)
      \bigr], \\
    \cI_m(\Gamma)
    &\coloneqq \Gamma \cup V_m(\Gamma) \cup W_m(\Gamma),
  \end{split}
\end{align}
where $x_1, \dots, x_{t(\Gamma)}$ are the extension variables used in~$\Lambda$.

To prove~$\cI_m(\Gamma)$ in~$\GER^-$ by simulating the $\ER$~proof of~$\Gamma$,
we essentially remove the clauses in~$V_m(\Gamma)$,
derive the extension clauses, and rederive~$V_m(\Gamma)$
by a sequence of blocked clause additions.

\begin{lemma}\label{thm:guarding-blocked-pair-easy-for-GER}
  For every formula~$\Gamma$ and every~$m \in \NN^+$,
  $\size_{\GER^-}(\cI_m(\Gamma)) \leq \size_{\ER}(\Gamma)$.
\end{lemma}
\begin{proof}
  Let $(\Lambda, \Pi)$ be the minimum-size $\ER$~proof of~$\Gamma$.
  We will show that the clauses in~$\Lambda \cup V_m(\Gamma)$
  can be derived from~$\Gamma \cup W_m(\Gamma)$
  in some sequence by blocked clause additions,
  which implies by \cref{def:blocked-extension}
  that $\Lambda$ is a blocked extension for~$\cI_m(\Gamma)$.

  Recall that extension clauses can be derived in sequence
  by blocked clause additions.
  The formula~$\Lambda$ is an extension for~$\Gamma \cup W_m(\Gamma)$,
  so we derive~$\Lambda$ by such a sequence.
  Next, from~$\Gamma \cup W_m(\Gamma) \cup \Lambda$,
  we derive the clauses in~$V_m(\Gamma)$ in any order.
  Let $V'$ be a proper subset of~$V_m(\Gamma)$,
  and let $C$ be a clause in~$V_m(\Gamma) \setminus V'$.
  For some $i \in [t(\Gamma)]$ and $j \in [m]$,
  the clause~$C$ is of the form~$p \lor y_j \lor \lneg{z_j}$,
  where $p$ is either $x_i$~or~$\lneg{x_i}$.
  With respect to~$\Gamma \cup W_m(\Gamma) \cup \Lambda \cup V'$,
  the clause~$C$ is blocked for~$y_j$
  since the only earlier occurrence of~$\lneg{y_j}$
  is the clause~$\lneg{y_j} \lor z_j$
  and $\{p,\, \lneg{z_j},\, z_j\}$ is tautological.
  It follows by induction that we can derive~$V_m(\Gamma)$
  from~$\Gamma \cup W_m(\Gamma) \cup \Lambda$.
  Thus, $\Lambda$ is a blocked extension for~$\cI_m(\Gamma)$.

  Noting that $\Pi$ is a resolution proof of~$\Gamma \cup \Lambda$
  and that $\cI_m(\Gamma)$ contains~$\Gamma$ as a subset,
  we conclude that there exists a $\GER^-$~proof of~$\cI_m(\Gamma)$
  of size~$\abs{\Lambda} + \abs{\Pi} = \size_\ER(\Gamma)$.
\end{proof}

For $\SBC^-$, the formula~$\cI_m(\Gamma)$
stays at least as hard as~$\Gamma$
if fewer than $2^m$~set-blocked clauses are derived.
As before, our goal is to render the added variables
useless in set-blocked clause additions.

We will eventually choose~$\Gamma$ to be hard for~$\SBC^-$,
which makes the literals $\lneg{y_j}$~and~$z_j$
useless in set-blocked clause additions.
Moreover, the presence of the clause~$\lneg{y_j} \lor z_j$ ensures that
if a clause is set-blocked for a set containing $y_j$~or~$\lneg{z_j}$,
then the clause is a weakening of~$y_j \lor \lneg{z_j}$.
Such clauses are killed by assignments
that set $y_j$~and~$z_j$ to the same value.

We also need to consider the clauses that are set-blocked
for sets containing the variables~$x_i$.
The projection of~$\cI_m(\Gamma)$ onto $x_i$~or~$\lneg{x_i}$
is the formula~$\bigcup_{j = 1}^m \{y_j \lor \lneg{z_j}\}$,
which has $2^m$~minimal satisfying assignments.
Without deriving clauses that rule out all of those assignments,
$\SBC^-$~proofs cannot use the variables~$x_i$ in any meaningful way.
Since the variables $y_j$~and~$z_j$
are rendered useless in set-blocked clause additions,
the assignments can only be ruled out one at a time,
which forces $\SBC^-$~proofs of~$\cI_m(\Gamma)$
to either derive at least~$2^m$ clauses or ignore the variables~$x_i$.

\begin{lemma}\label{thm:guarding-blocked-pair-hard-for-SBC}
  For every formula~$\Gamma$ and every~$m \in \NN^+$,
  \begin{equation*}
    \size_{\SBC^-}(\cI_m(\Gamma)) \geq \min\{2^m, \size_{\SBC^-}(\Gamma)\}.
  \end{equation*}
\end{lemma}
\begin{proof}
  Fix some~$m \in \NN^+$.
  When $\Gamma$ is satisfiable, the inequality holds trivially,
  so suppose that $\Gamma$ is unsatisfiable.

  Suppose that $\cI_m(\Gamma)$ has an $\SBC^-$~proof of size~$N$.
  By \cref{thm:WLOG-SBC-preprocesses}, we may view such a proof
  as a resolution proof of the formula~$\cI_m(\Gamma) \cup \Sigma$,
  where $\Sigma$ is a set of clauses derivable from~$\cI_m(\Gamma)$
  by a sequence of set-blocked clause additions without new variables.
  We claim that if $\abs{\Sigma} < 2^m$,
  then there exists an $\SBC^-$~proof of~$\Gamma$ of size at most~$N$.
  This implies the desired lower bound.

  Let $X = \left\{x_1, \dots, x_{t(\Gamma)}\right\}$
  and $U = \left\{y_1, \dots, y_m, z_1, \dots, z_m\right\}$
  denote the two sets of variables incorporated into~$\cI_m(\Gamma)$.
  Consider a clause~$C = L \ldor C'$ that is set-blocked
  for~$L$ with respect to~$\cI_m(\Gamma)$.
  We start by inspecting the ways in which
  the variables from~$X \cup U$ can occur in~$L$:

  \begin{itemize}
  \item We first consider the variables from~$U$.
    If either $\lneg{y_j}$~or~$z_j$ for some~$j \in [m]$ occurs in~$L$,
    then, by \cref{thm:SBC-projection-onto-literal},
    the assignment~$L \cup \lneg{C'}$ satisfies~$\Gamma$
    since $\Gamma$ is contained in the projections
    of~$\cI_m(\Gamma)$ onto the negations of these literals.
    Thus, since $\Gamma$ is unsatisfiable,
    neither of the literals $\lneg{y_j}$~and~$z_j$
    for any~$j \in [m]$ can occur in~$L$.
    Moreover, if the literal~$y_j$ occurs in~$L$,
    then, by \cref{thm:SBC-projection-onto-literal},
    the assignment~$L \cup \lneg{C'}$ satisfies~$z_j$
    since the clause~$\lneg{y_j} \lor z_j$ is in~$\cI_m(\Gamma)$.
    In that case, since $z_j$ cannot occur in~$L$,
    the literal~$z_j$ must occur in~$\lneg{C'}$,
    which implies that $C$ contains the literal~$\lneg{z_j}$.
    Thus, if the literal~$y_j$ occurs in~$L$,
    then $C$ contains the literal~$\lneg{z_j}$.
    By a similar argument, if the literal~$\lneg{z_j}$ occurs in~$L$,
    then $C$ contains the literal~$y_j$.
    To summarize, if some variable from~$U$ is in~$L$,
    then $C$ is a weakening of some clause
    in~$\bigcup_{j = 1}^m \{y_j \lor \lneg{z_j}\}$.
  \item Next, we consider the variables from~$X$.
    For all~$j \in [m]$, define $A_j \coloneqq \{y_j,\, \lneg{z_j}\}$
    (intended to be viewed as an assignment).
    Let $\bA = A_1 \times \dots \times A_m$.
    Suppose that a literal~$p$ of some~$x_i$ is in~$L$.
    Then, by \cref{thm:SBC-projection-onto-literal},
    the assignment~$L \cup \lneg{C'}$ satisfies the formula
    \begin{equation*}
      \proj_{\lneg{p}}(\cI_m(\Gamma)) = \bigcup_{j = 1}^m \{y_j \lor \lneg{z_j}\}.
    \end{equation*}
    This implies that if no variable from~$U$ occurs in~$L$,
    then there exists some assignment~$\beta \in \bA$
    such that $\beta \subseteq \lneg{C'}$.
    We say $C$~is a \emph{good} clause
    if some variable from~$X$ occurs in~$L$
    but no variable from~$U$ occurs in~$L$.
  \end{itemize}

  From this point on, suppose $\abs{\Sigma} < 2^m$.
  For each good clause~$E$ in~$\Sigma$,
  choose a single subset~$F \subseteq E$ such that $\lneg{F} \in \bA$.
  Let $\Delta$ be the collection of those subsets.
  Since $\abs{\Delta} < 2^m$, there exists some~$\beta \in \bA$
  such that $\lneg{\beta} \notin \Delta$.
  Recall that for each~$j \in [m]$, the assignment~$\beta$
  sets exactly one of the variables $y_j$~and~$z_j$.
  Let $\beta'$ be the smallest assignment extending~$\beta$
  such that $\beta'(y_j) = \beta'(z_j)$ for all~$j \in [m]$.

  \begin{claim}\label{thm:beta'-satisfies-XU-SBCs}
    Let $C$ be a clause in~$\Sigma$,
    and let $L \subseteq C$ be a subset
    for which $C$ is set-blocked with respect to~$\cI_m(\Gamma)$.
    If some variable from~$X \cup U$ occurs in~$L$,
    then $\beta'$ satisfies~$C$.
  \end{claim}
  \begin{proof}
    Let $C$ be a clause in~$\Sigma$
    set-blocked for~$L \subseteq C$ with respect to~$\cI_m(\Gamma)$.
    Suppose that some variable from~$X \cup U$ occurs in~$L$.
    Then either $\var(L) \cap U \neq \varnothing$ or $C$ is a good clause.
    \begin{description}
    \item[Case 1:] $\var(L) \cap U \neq \varnothing$.
      Since $C$ is a weakening of some clause
      in~$\bigcup_{j = 1}^m \{y_j \lor \lneg{z_j}\}$
      and $\beta'(y_j) = \beta'(z_j)$ for all~$j \in [m]$,
      the assignment~$\beta'$ satisfies~$C$.
    \item[Case 2:] \textit{$C$ is a good clause.}
      Let $F$ be a subset of~$C$ such that $F \in \Delta$.
      Since $\lneg{\beta} \notin \Delta$,
      there exists some~$j \in [m]$ such that
      either $\lneg{y_j} \in \lneg{\beta}$~and~$z_j \in F$
      or $z_j \in \lneg{\beta}$~and~$\lneg{y_j} \in F$.
      We have $\beta'(z_j) = 1$ in the former case
      and $\beta'(y_j) = 0$ in the latter.
      Either way, $\beta'$ satisfies~$F$ and hence it also satisfies~$C$. \qedhere
    \end{description}
  \end{proof}

  Now, let $\alpha$ be the assignment defined as
  \begin{equation*}
    \alpha(v) =
    \begin{cases}
      1 & \text{if } v \in X \\
      \beta'(v) & \text{if } v \in U \\
      \text{undefined} & \text{otherwise}.
    \end{cases}
  \end{equation*}
  The point of~$\alpha$ is to set all of the variables from~$X \cup U$
  in such a way that kills all of the clauses in~$\Sigma$
  that are set-blocked for sets containing those variables,
  leaving behind only the clauses that could
  also be derived in an $\SBC^-$~proof of~$\Gamma$.

  The rest of the argument is similar at a high level
  to the proof of \cref{thm:projection-hard-for-SBC}, so we will be relatively brief.
  By \cref{thm:resolution-under-restrictions}, there exists a resolution proof
  of the formula~$(\cI_m(\Gamma) \cup \Sigma)|_\alpha = \Gamma \cup \Sigma|_\alpha$
  of size at most~$N - \abs{\Sigma}$.
  We claim that the clauses in~$\Sigma|_\alpha$ can be derived
  in sequence from~$\Gamma$ by set-blocked clause additions.

  As before, let $S = (C_1, \dots, C_r)$ be the ordering in which
  the clauses of~$\Sigma$ are derived from~$\cI_m(\Gamma)$.
  We will prove that for all~$i \in [r]$
  such that $\alpha$ does not satisfy~$C_i$,
  the clause~$C_i|_\alpha$ is set-blocked
  with respect to~$\Gamma \cup \Phi_{i - 1}|_\alpha$, where
  \begin{equation*}
    \Phi_{i - 1} \coloneqq \bigcup_{j \in [i - 1]} \{C_j\}.
  \end{equation*}

  Let $i \in [r]$, and consider the clause~$C_i$,
  which we write as~$C$ from this point on.
  Let $L \subseteq C$ be a subset for which $C$ is set-blocked
  with respect to~$\cI_m(\Gamma) \cup \Phi_{i - 1}$.
  Suppose that $\alpha$ does not satisfy~$C$.
  Noting that $\alpha$ extends~$\beta'$,
  by \cref{thm:beta'-satisfies-XU-SBCs},
  no variable from~$X \cup U$ is in~$L$.
  As a result, $L|_\alpha = L$.
  We will prove that $C|_\alpha$ is set-blocked for~$L$
  with respect to both $\Gamma$~and~$\Phi_{i - 1}|_\alpha$.

  \begin{description}
  \item\hspace{-\labelsep}\ignorespaces
    \textit{$C|_\alpha$ is set-blocked
      for~$L$ with respect to~$\Gamma$}:
    Since $\Gamma \subseteq \cI_m(\Gamma)$, the clause~$C$
    is set-blocked for~$L$ in particular with respect to~$\Gamma$.
    No variable from~$X \cup U$ occurs in~$\Gamma$,
    so the clause~$C|_\alpha$ also is set-blocked
    for~$L$ with respect to~$\Gamma$.
  \item\hspace{-\labelsep}\ignorespaces
    \textit{$C|_\alpha$ is set-blocked
      for~$L$ with respect to~$\Phi_{i - 1}|_\alpha$}:
    Consider an arbitrary~$D' \in \Phi_{i - 1}|_\alpha$,
    which is the restriction under~$\alpha$
    of some clause~$D \in \Phi_{i - 1}$ that $\alpha$ does not satisfy.
    Suppose $D' \cap \lneg{L} \neq \varnothing$ and $D' \cap L = \varnothing$.
    We need to show that $(C|_\alpha \setminus L)
    \cup (D' \setminus \lneg{L})$ is tautological.

    We have $D \cap \lneg{L} \neq \varnothing$ because $D' \subseteq D$.
    Recall that no variable from~$X \cup U$ is in~$L$,
    and observe that $D'$ is simply $D$~with the variables from~$X \cup U$ removed.
    This implies $D \cap L = \varnothing$.
    Now, because $C$ is set-blocked for~$L$ with respect to~$\Phi_{i - 1}$,
    the set~$E = (C \setminus L) \cup (D \setminus \lneg{L})$ must be tautological.
    A variable that occurs both positively and negatively in~$E$
    cannot be from~$X \cup U$, since in that case $\alpha$ would satisfy $C$~or~$D$.
    Therefore, the set~$(C|_\alpha \setminus L)
    \cup (D' \setminus \lneg{L})$ also is tautological. \qedhere
  \end{description}
\end{proof}

Invoking \crefnosort{thm:guarding-blocked-pair-easy-for-GER,%
  thm:guarding-blocked-pair-hard-for-SBC} with a suitable choice of~$m$
and with $\Gamma$~as the bit pigeonhole principle gives us the separation.

\begin{theorem}\label{thm:separation-GER-from-SBC}
  For every unsatisfiable formula~$\Gamma$,
  let $m(\Gamma) \coloneqq \ceil{\log(\size_{\SBC^-}(\Gamma))}$
  and define $\cK(\Gamma) \coloneqq \cI_{m(\Gamma)}(\Gamma)$.
  The formula~$\cK(\BPHP_n)$ admits polynomial-size proofs in~$\GER^-$
  but requires exponential-size proofs in~$\SBC^-$.
\end{theorem}
\begin{proof}
  Buss and Thapen~\cite[Theorem~4.4]{BT21}
  gave polynomial-size proofs of~$\BPHP_n$
  in~$\SPR^-$, which $\ER$ simulates.
  By \cref{thm:guarding-blocked-pair-easy-for-GER}, we have
  \begin{equation*}
    \size_{\GER^-}(\cK(\BPHP_n)) = n^{O(1)}.
  \end{equation*}
  \crefnosort{thm:SBC-size-lower-bound-for-BPHP,%
    thm:guarding-blocked-pair-hard-for-SBC} give
  \begin{equation*}
    \size_{\SBC^-}(\cK(\BPHP_n)) = 2^{\Omega(n)}.
  \end{equation*}
  Thus, the bit pigeonhole principle with $\cK$~applied to it
  exponentially separates~$\GER^-$ from~$\SBC^-$.
\end{proof}

\section*{Acknowledgments}
\label{sec:acknowledgments}

I thank Sam Buss, Marijn Heule, Jakob Nordström, and Ryan O'Donnell
for useful discussions.
I thank Jeremy Avigad, Ryan O'Donnell, and Bernardo Subercaseaux
for feedback on an earlier version of the paper.
Finally, I thank the STACS reviewers
for their highly detailed comments and suggestions.

\newcommand{\Proc}[1]{Proceedings of the \nth{#1}}
\newcommand{\STOC}[1]{\Proc{#1} Symposium on Theory of Computing (STOC)}
\newcommand{\FOCS}[1]{\Proc{#1} Symposium on Foundations of Computer Science
  (FOCS)} \newcommand{\SODA}[1]{\Proc{#1} Symposium on Discrete Algorithms
  (SODA)} \newcommand{\CCC}[1]{\Proc{#1} Computational Complexity Conference
  (CCC)} \newcommand{\ITCS}[1]{\Proc{#1} Innovations in Theoretical Computer
  Science (ITCS)} \newcommand{\ICS}[1]{\Proc{#1} Innovations in Computer
  Science (ICS)} \newcommand{\ICALP}[1]{\Proc{#1} International Colloquium on
  Automata, Languages, and Programming (ICALP)}
\newcommand{\STACS}[1]{\Proc{#1} Symposium on Theoretical Aspects of Computer
  Science (STACS)} \newcommand{\MFCS}[1]{\Proc{#1} International Symposium on
  Mathematical Foundations of Computer Science (MFCS)}
\newcommand{\LICS}[1]{\Proc{#1} Symposium on Logic in Computer Science
  (LICS)} \newcommand{\CSL}[1]{\Proc{#1} Conference on Computer Science Logic
  (CSL)} \newcommand{\CSLw}[1]{\Proc{#1} International Workshop on Computer
  Science Logic (CSL)} \newcommand{\DLT}[1]{\Proc{#1} International Conference
  on Developments in Language Theory (DLT)} \newcommand{\TAMC}[1]{\Proc{#1}
  International Conference on Theory and Applications of Models of Computation
  (TAMC)} \newcommand{\RTA}[1]{\Proc{#1} International Conference on Rewriting
  Techniques and Applications (RTA)} \newcommand{\IJCAR}[1]{\Proc{#1}
  International Joint Conference on Automated Reasoning (IJCAR)}
\newcommand{\CADE}[1]{\Proc{#1} Conference on Automated Deduction (CADE)}
\newcommand{\SAT}[1]{\Proc{#1} International Conference on Theory and
  Applications of Satisfiability Testing (SAT)}
\newcommand{\TABLEAUX}[1]{\Proc{#1} International Conference on Automated
  Reasoning with Analytic Tableaux and Related Methods (TABLEAUX)}
\newcommand{\TACAS}[1]{\Proc{#1} International Conference on Tools and
  Algorithms for the Construction and Analysis of Systems (TACAS)}
\newcommand{\LPAR}[1]{\Proc{#1} International Conference on Logic for
  Programming, Artificial Intelligence and Reasoning (LPAR)}
\newcommand{\HVC}[1]{\Proc{#1} Haifa Verification Conference (HVC)}
\newcommand{\DAC}[1]{\Proc{#1} Design Automation Conference (DAC)}
\newcommand{\DATE}{Proceedings of the Design, Automation and Test in Europe
  Conference (DATE)} \newcommand{\ISAIM}[1]{\Proc{#1} International Symposium
  on Artificial Intelligence and Mathematics (ISAIM)}
\newcommand{\AAAI}[1]{\Proc{#1} AAAI Conference on Artificial Intelligence
  (AAAI)} \newcommand{\ICML}[1]{\Proc{#1} International Conference on Machine
  Learning (ICML)} \newcommand{\NeurIPS}[1]{\Proc{#1} Conference on Neural
  Information Processing Systems (NeurIPS)}

\end{document}